\documentclass[10pt,onecolumn,twoside] {IEEEtran}

\usepackage{graphicx}
\usepackage{amsmath}
\usepackage{amsfonts}
\usepackage{amssymb}
\usepackage{enumerate}
\usepackage{lscape}
\usepackage{longtable}
\usepackage{url}

\newtheorem{theorem}{Theorem}[section]

\newtheorem{corollary}{Corollary}[section]

\newtheorem{definition}{Definition}[section]

\newtheorem{proposition}{Proposition}[section]

\newenvironment{proof}[1][Proof]{\textbf{#1.} }{\ \rule{0.5em}{0.5em} \vspace{1ex}}

\newcommand{\OOO}{\mathcal{O}}
\newcommand{\RRR}{\mathcal{R}}

\newcommand{\NNN}{\mathcal{N}}

\newcommand{\vad}{u}%as variaveis em R^n
\newcommand{\RR}{\mathbb{R}}

\begin{document}

\title{On Partial Sparse Recovery}
\author{
A. S. Bandeira, K. Scheinberg, L. N. Vicente%
\thanks{ASB is with the Program in Applied and Computational Mathematics, Princeton University, Princeton, NJ 08544, USA ({\tt ajsb@math.princeton.edu}).}%
\thanks{KS is with the
Department of Industrial and Systems Engineering, Lehigh University,
Harold S. Mohler Laboratory, 200 West Packer Avenue, Bethlehem, PA 18015-1582, USA
({\tt katyas@lehigh.edu}). Support for this author is provided by AFOSR  under grant  FA9550-11-1-0239 and by NSF under grant DMS 10-16571.}%
\thanks{LNV is with CMUC, Department of Mathematics, University of Coimbra,
3001-454 Coimbra, Portugal ({\tt lnv@mat.uc.pt}). Support for
this author was provided by FCT under grant PTDC/MAT/098214/2008.}%
}

\markboth{}%
{Bandeira \MakeLowercase{\textit{et al.}}: On Partial Sparse Recovery}

%\date{}
\maketitle

%\footnotesep=0.4cm

\begin{abstract}
We consider the problem of recovering a partially sparse solution  of
an underdetermined system of linear equations by  minimizing the $\ell_1$-norm
of the part of the solution vector which is known to be sparse.
Such a problem is closely related to a classical problem
in Compressed Sensing where the  $\ell_1$-norm of the whole solution vector is minimized.
We introduce analogues of  restricted isometry and null space properties  for
the recovery of partially sparse
vectors and show that these new properties are  implied by their original counterparts.
We show also how to extend recovery under noisy measurements to the partially sparse case.
\end{abstract}

\begin{IEEEkeywords}
Partial sparse recovery, compressed sensing, $\ell_1$-minimization,
Sparse quadratic polynomial interpolation.
\end{IEEEkeywords}

%\IEEEpeerreviewmaketitle

%%%%%%%%%%%%%%%%%%%%%%%%%%%%%%%%%%%%%%%%%%%%%%%%%%%%%%%%%%%%%%%%%%%%%%%

\section{Introduction}\label{sec:introduction}

\IEEEPARstart{I}{n} Compressed Sensing one is interested in recovering a sparse
solution~$\bar x\in\RR^N$ of an underdetermined system of the form  $y=A \bar x$, given
a vector $y\in\RR^k$ and a matrix $A \in \RR^{k\times N}$
with far fewer rows than columns $(k\ll N)$.
A direct approach  is to  minimize the number of non-zero components of~$x$, i.e.,
the $\ell_0$-\emph{norm} of $x$ (which is defined
as $\|\vad\|_0=|\{i:\vad_i\neq 0\}|$ but, strictly speaking, is not a norm),
\begin{equation}\label{minl0fourier}
\min \|x\|_0 \quad \operatorname{s.t.}\quad A x= y.
\end{equation}
Since (\ref{minl0fourier}) is known to be NP-Hard, a tractable approximation  is commonly considered which is  obtained by  substituting the
 non-convex $\ell_0$-\emph{norm} by a convex approximation.
Recent results indicate that the $\ell_1$-norm can serve as such an approximation
(see~\cite{ECandes_2006} for a survey on some of this material). % and, in fact,
Hence (\ref{minl0fourier}) is replaced
by the following optimization problem
\begin{equation}\label{minl1}
\min \|x\|_1 \quad \operatorname{s.t.}\quad Ax=y.
\end{equation}
Note that~(\ref{minl1}) is equivalent to a linear program and thus is much
easier to solve than~(\ref{minl0fourier}).

In this paper we consider the case
(see~\cite{NVaswani_WLu_2011,MPFriedlander_etal_2010,LJacques_2010})
 when it is known a priori that the solution vector consists of two parts,
one of which is expected to be dense, in other words
we have $x=(x_1,x_2)$, where $x_1 \in\RR^{N-r}$
is sparse and $x_2 \in\RR^r$ is possibly dense.
A natural generalization of problem (\ref{minl1}) to this setting of
partially sparse recovery is given by
\begin{eqnarray}
\min \|x_1\|_1 & \operatorname{s.t.} & A_1 x_1 + A_2 x_2 = y,\label{minpartiall1}
\end{eqnarray}
where $A=(A_1, A_2)$, $A_1 \in \RR^{k \times (N-r)}$, and
$A_2 \in \RR^{k \times r}$. We will refer to this setting as {\em partially sparse recovery of size $N-r$}.
One of the key applications of partially sparse recovery is image reconstruction~\cite{NVaswani_WLu_2011}
but they also arise naturally in sparse
Hessian recovery~\cite{ABandeira_KScheinberg_LNVicente_2010}.

Vaswani and Lu~\cite{NVaswani_WLu_2011} gave a first sufficient
condition for partially sparse recovery.
Later, Friedlander et al.~\cite{MPFriedlander_etal_2010} proposed a weaker sufficient condition
and covered the extension to the noisy case.
After obtaining our results we were directed to the work of Jacques~\cite{LJacques_2010} who addressed the noisy case, deriving another sufficient condition for partially sparse recovery. His conditions guarantee the same recovery as ours but, as far as we can tell, are not the simple extensions of the NSP and RIP properties.
The conditions in~\cite{NVaswani_WLu_2011,MPFriedlander_etal_2010,LJacques_2010} are
somewhat weaker than the known
restricted isometry property for general sparse recovery, which is natural since the case of partial sparsity can be considered as a case of general sparsity where part of the support of the solution is known in advance.

The contribution of our paper is to introduce
the analogues of restricted isometry and null space properties for the case of partial sparsity.
We prove that these new properties are sufficient for partially sparse recovery
(including  the noisy case) and are implied by the original conditions of fully sparse
recovery. We show that it is possible to guarantee recovery of a partially sparse signal using 
Gaussian random matrices with the number of measurements an order smaller than the one necessary for general recovery.

%Under the restricted isometry property for partially sparse recovery,
%we show that it is possible to guarantee sparse recovery (with Gaussian random matrices)
%for an order of measurements below the one necessary for general recovery.

\subsection{Notation}
We will use the following notation in this paper. $[N]$ denotes the set of integers
 $\{1, \ldots, N\}$, and $[N]^{(s)}$ denotes the set of all subsets of $[N]$ of cardinality $s\leq N$. If $A$ is a matrix, then by $\NNN(A)$ and  $\RRR(A)$ we denote the null and range spaces of $A$, respectively. We say that
 a vector $x$ is $s-$sparse if at most $s$ components of $x$ are non-zero. This is also denoted by  $\|x\|_0\leq s$. Given $v\in\RR^N$ and $S\in[N]$, $v_S\in\RR^N$ denotes a vector defined by $(v_S)_i=v_i$, $i\in S$ and $(v_S)_i=0$, $i\notin S$.

\section{Sparse recovery in compressed sensing}\label{sec:CS}

One of the main questions addressed by Compressed Sensing is under what conditions on
the matrix~$A$ can every sparse vector $\bar x$ be recovered  by solving problem~(\ref{minl1}) given~$A$ and the
right hand side  $y=A \bar x$.
The next definition is a well known characterization of such
matrices (see, e.g., \cite{ACohen_WDahmen_RDeVore_2009,DDonoho_XHuo_2001}).

\begin{definition}[Null Space Property]\label{def:NSP}
The matrix $A\in\RR^{k\times N}$ is said to satisfy the Null Space Property
(NSP) of order $s$ if, for every $v\in \NNN(A) \setminus \{0\}$ and for
every $S~\in~[N]^{(s)}$, one has
\begin{equation}\label{vsmenor12v}
\|v_S\|_1 \; < \; \frac12\|v\|_1.
\end{equation}
\end{definition}

It is well known that NSP is a necessary and sufficient condition for the recovery of an $s$-sparse vector $\bar x$ (see \cite{HRauhut_2010}).

\begin{theorem}\label{teoremaNSP}
The matrix $A$ satisfies the Null Space Property of order $s$ if and only if,
for every $s-$sparse vector $\bar x$, problem~(\ref{minl1}) with $y=A \bar x$
has an
unique solution and it is given by $x = \bar x$.
\end{theorem}

It is difficult to analyze whether NSP is satisfied. On the other hand, the
\emph{Restricted Isometry Property} (RIP), introduced in~\cite{ECandes_TTao_2006},
is considerably more
useful and insightful, although it provides only sufficient conditions for recovery with~(\ref{minl1}).
We present below the definition of the
\emph{RIP Constant}.

\begin{definition}[Restricted Isometry Property Constant]\label{defRIP}
One says that $\delta_s>0$ is the Restricted Isometry Property Constant,
or \emph{RIP} constant, of order $s$ of the matrix
$A\in\RR^{k\times N}$ if $\delta_s$ is the smallest positive real number such that:
\begin{equation}\label{defdeltasRIP}\left(1-\delta_s\right)\|x\|_2^2 \; \leq \;
\|Ax\|_2^2 \; \leq \; \left(1+\delta_s\right)\|x\|_2^2\end{equation}
for every $s-$sparse vector $x$.
\end{definition}

The following theorem (see, e.g.,~\cite{EJCandes_2009}) provides a useful
sufficient condition for successful recovery by~(\ref{minl1}).

\begin{theorem}\label{teoremaRIP} \cite{EJCandes_2009}
Let $A\in\RR^{k\times N}$ and $2s < k$. If
$\delta_{2s}<\sqrt{2}-1$,
 where $\delta_{2s}$ is the RIP constant of $A$ of order
 $2s$, then, for every $s-$sparse vector
$\bar x$, problem~(\ref{minl1}) with $y=A \bar x$ has an unique solution
and it is given by $x = \bar x$.
\end{theorem}

It is known that RIP is satisfied with some probability if the
entries of the matrix are randomly generated
(see, e.g.,~\cite{RBaraniuk_MDavenport_RDeVore_MWakin_2008}) according
to some distribution such as a sub-Gaussian. However, it is in general computationally hard to check whether
it is satisfied by a certain realization matrix~\cite{Bandeira_etal_hardRIP},
and it is still an open problem to find such matrices deterministically when the underlying
system is highly underdetermined (see~\cite{Bandeira_etal_FlatRIP}).

\section{Partial sparse recovery}\label{sec:partialCS}

In this section we consider the following extension of the NSP
to the case of partially sparse recovery.

\begin{definition}[Partial Null Space Property]\label{def:NSPpartial}
We say that $A=(A_1, A_2)$ satisfies the Null Space Property (NSP) of order $s-r$ for partially sparse recovery of size $N-r$ with $r\leq s$
if $A_2$ is full column rank (${\cal N}(A_2)=\{0\}$) and for every $v_1\in\RR^{N-r}\setminus \{0\}$ such that $A_1v_1\in \RRR(A_2)$ and every $S\in[N-r]^{(s-r)}$, we have
\begin{equation}\label{def:NSPpartialigualdade1}
\|(v_1)_S\|_1 \; < \; \frac12\|v_1\|_1.
\end{equation}
\end{definition}

Note that when $r=0$, the partial NSP naturally reduces to the
NSP in Definition~\ref{def:NSP}.
Wang and Yin~\cite{supportdetection} have suggested a stronger NSP
adapted to a setting where it is not known the location of the
partial support.

The new property is a necessary and
sufficient condition  for any solution of (\ref{minpartiall1})
with $y=A\bar x$ to satisfy $x=\bar x$ if $\bar x_1$ is
appropriately sparse.

\begin{theorem}\label{teoremaNSPpartial}
The matrix $A=(A_1, A_2)$ satisfies the Null Space Property of order $s-r$
for Partially Sparse Recovery of size $N-r$ if and only if for every
$\bar x=(\bar x_1,\bar x_2)$ such that $\bar x_1\in\RR^{N-r}$ is $(s-r)-$sparse and
$\bar x_2\in\RR^r$,
problem~(\ref{minpartiall1}) with  $y=A \bar x$ has an unique solution and it is given by
$(x_1,x_2)= (\bar x_1,\bar x_2)$.
\end{theorem}

\begin{proof}
The proof follows the steps of the proof of \cite[Theorem~2.3]{HRauhut_2010} with appropriate
modifications. Let us assume first that for any vector
$(\bar x_1,\bar x_2)\in\RR^N$, where $\bar x_1$ is an $(s-r)-$sparse
vector and $\bar x_2\in\RR^r$, the minimizer $(x_1,x_2)$ of
 $\|x_1\|_1$ subject to
$A_1 x_1+A_2 x_2 = A \bar x$ satisfies $x_1= \bar x_1$.
Consider any $v_1 \neq 0$ such that $A_1v_1\in\RRR(A_2)$. Then consider
 minimizing  $\|x_1\|_1$ subject to $A_1 x_1+A_2 x_2=A_1 (v_1)_S+A_2 v_2$ for any
$v_2 \in\RR^r$ and for any $S\in[N-r]^{(s-r)}$. By the assumption, the corresponding minimizer $(x_1,x_2)$
 satisfies $x_1=(v_1)_S$.
Since $A_1v_1\in\RRR(A_2)$, there exists $ u_2$ such that  $A_1 (-(v_1)_{S^c}) + A_2  u_2 = A_1 (v_1)_S + A_2 v_2$. As $-(v_1)_{S^c}\neq (v_1)_S$,
$(-(v_1)_{S^c}, u_2)$ is not the minimizer of $\| x_ 1 \|_1$ subject to
$A_1 x_1 + A_2 x_2 = A_1 (v_1)_S + A_2 v_2$, hence, $\|(v_1)_{S^c}\|_1>\|(v_1)_{S}\|_1$
and~(\ref{def:NSPpartialigualdade1}) holds.

Let us now assume that $A$ satisfies the NSP of order $s-r$
for partially sparse recovery of size $N-r$ (Definition \ref{def:NSPpartial}).
Then, given a vector $(\bar x_1, \bar x_2)\in\RR^N$, where $\bar x_1$ is $(s-r)-$sparse and $\bar x_2\in\RR^r$, and a vector $(u_1,u_2)\in\RR^N$ with
$u_1 \neq \bar x_1$ and satisfying $A_1 u_1 + A_2 u_2 = A_1 \bar x_1 + A_2 \bar x_2$, consider
$(v_1,v_2)=\left((\bar x_1-u_1),(\bar x_2-u_2)\right)\in \NNN(A)$, which implies $A_1v_1\in\RRR(A_2)$
and $v_1\neq 0$.
Thus, setting $S$ to be the support of $\bar x$, one has that
\begin{eqnarray*}
\|\bar x_1\|_1&\leq&\|\bar x_1-(u_1)_S\|_1+\|(u_1)_S\|_1\\[1ex]
&=&\|(v_1)_S\|_1+\|(u_1)_S\|_1
<\|(v_1)_{S^c}\|_1+\|(u_1)_S\|_1\\ [1ex]
&=&\|-(u_1)_{S^c}\|_1+\|(u_1)_S\|_1
=\|u_1\|_1,
\end{eqnarray*}
(the strict inequality coming from (\ref{def:NSPpartialigualdade1})),
guaranteeing that all solutions $(x_1,x_2)$ of~(\ref{minpartiall1}) with $y=A \bar x$ satisfy $x_1= \bar x_1$.

It remains to note that $x_2 = \bar x_2$ is
uniquely determined by solving
$A_2 x_2=y-A_1 \bar x_1$ if and only if $A_2$ is full column rank.
\end{proof}

We now define an extension of the RIP to the partially sparse
recovery setting.
For this purpose, let $A=(A_1, A_2)$ be as considered above, under the
 assumption that $A_2$ has full column rank. Let
\begin{equation}\label{projM}
{\cal P} \; = \; I-A_2\left(A_2^\top A_2\right)^{-1}A_2^\top
\end{equation}
be the matrix of the orthogonal projection from $\RR^N$ onto
$\RRR\left(A_2\right)^\bot.$ Then, the problem of recovering
$(\bar x_1, \bar x_2)$, where $\bar x_1$ is an $(s-r)-$sparse vector satisfying
$A_1 \bar x_1+A_2 \bar x_2=y$, can be stated as the problem of recovering an
$(s-r)-$sparse vector $x_1= \bar x_1$ satisfying $\left({\cal P}A_1\right)x_1={\cal P}y$
and then recovering $x_2 = \bar x_2$ satisfying $A_2x_2=y-A_1\bar x_1$.
The solution of the resulting linear system in the second step exists and is unique
  given that
$A_2$ has full column rank and $({\cal P}A_1) \bar x_1={\cal P}y$. Note that
the first step is now reduced to the classical setting of Compressed Sensing.
This motivates the following definition of RIP for partially sparse recovery.

\begin{definition}[Partial RIP]\label{defpartialRIP}
We say that $\delta_{s-r}^r>0$ is the Partial Restricted Isometry
Property Constant of order $s-r$  of the
matrix $A=(A_1, A_2)\in\RR^{k\times N}$, for recovery of size $N-r$ with $r\leq s$, if
$A_2$ is full column rank and
$\delta_{s-r}^r$
is the RIP constant of order $s-r$ (see Definition~\ref{defRIP}) of the
matrix ${\cal P}A_1$, where ${\cal P}$ is given by (\ref{projM}).
\end{definition}

Again, when $r=0$ the Partial RIP reduces to the RIP of
Definition~\ref{defRIP}. We also note that, given a matrix
$A=(A_1, A_2)\in\RR^{k\times N}$ with Partial RIP constant $\delta_{2(s-r)}^r$ of
order $2(s-r)$ for recovery of size $N-r$, satisfying
$\delta_{2(s-r)}^r<\sqrt{2}-1$,  Theorems~\ref{teoremaNSP} and~\ref{teoremaRIP},
guarantee that ${\cal P}A_1$ satisfies the NSP of order $s-r$. Thus,
given $\bar x=(\bar x_1, \bar x_2)$ such that $\bar x_1\in\RR^{N-r}$ is $(s-r)-$sparse
and $\bar x_2\in\RR^r$, $\bar x_1$ can be recovered by minimizing the
$\ell_1$-norm of $x_1$ subject to $({\cal P}A_1)x_1={\cal P}A \bar x$ and,
recalling that $A_2$ is full-column rank, $x_2 = \bar x_2$ is uniquely
determined by $A_2x_2=y-A_1 \bar x_1$. (In particular, this implies
that $A$ satisfies the NSP of order $s-r$ for partially sparse recovery
of size $N-r$.)

\section{Partially sparse recovery implied by fully sparse recovery conditions}\label{sec:partial-total}

We are now interested in showing that  partially sparse recovery is achievable
under the conditions which guarantee fully sparse recovery.
In particular we will show that the NSP and RIP imply, respectively, the partial NSP and the partial RIP.
We first establish the relationship between
 the corresponding null space properties.

\begin{theorem}\label{th:NSP-partial-total}
If a given matrix  $A$ satisfies the NSP of order $s$ then it satisfies
the  NSP for partially sparse recovery of order  $s-r$ for any $r\leq s$.
\end{theorem}

\begin{proof}
Let $A=(A_1,A_2)$ satisfy the
NSP of order $s$. First we note that since $r \leq s$, the NSP implies that $A_2$ is full column rank. Let $v_1\in\RR^{N-r}$ be a non-zero vector such that
 $A_1v_1\in\RRR(A_2)$
and let $T\in[N-r]^{(s-r)}$.

Since there exists $v_2$ such that $A_1v_1+A_2v_2=0$, we have that $v = (v_1,v_2) \in {\cal N}(A)\setminus \{0\}$,
and
therefore by setting $S=T\cup ([N]\setminus[N-r])$ and by using the NSP,
$\|(v_1)_T\|_1+\|v_2\|_1 = \|v_S\|_1 < \frac12\|v\|_1 = \frac12\|v_1\|_1+\frac12\|v_2\|_1.$
Thus,
\(\|(v_1)_T\|_1 \; \leq \; \|(v_1)_T\|_1+\frac12\|v_2\|_1 \; \leq \; \frac12\|v_1\|_1,\)
and $A$ satisfies the  NSP of order $s-r$ for partially
sparse recovery of size $N-r$.
\end{proof}

Partial RIP is also implied by  RIP without the change in the RIP constant value.

\begin{theorem}\label{th:RIP-partial-total}
Let $\delta_s>0$ and $A=(A_1,A_2)$ satisfy the following property:
For every $(s-r)$-sparse vector $x_1\in\RR^{N-r}$ and $x_2\in\RR^r$ we have
\begin{equation}\label{ineqThRIP}
(1-\delta_s)\|x\|_2^2 \leq \|Ax \|_2^2 \leq (1+\delta_s)\|x\|_2^2,
\end{equation}
where $x = (x_1, x_2)$. Then $A$ satisfies
 partial RIP of order $s-r$ with $\delta_{s-r}^r=\delta$ for partially sparse recovery of size $N-r$, for any  $r\leq s$.
\end{theorem}

\begin{proof}  First we note that setting $x_1=0$ implies that $A_2$ is full column rank. Consider now any given  $(s-r)-$sparse vector $x_1 \in \RR^{N-r}$.
Now, by setting $x_2 = - \left(A_2^\top A_2\right)^{-1}A_2^\top A_1 x_1$, one obtains
$(1-\delta_s) \|x_1\|_2^2 \leq
\left(1-\delta_s\right) \left( \|x_1\|_2^2+\|x_2\|_2^2 \right)
\leq \|A_1x_1 + A_2x_2\|_2^2  =  \| {\cal P} A_1 x_1 \|^2$.
On the other hand, the choice $x_2 = 0$ provides
$\|{\cal P} A_1x_1\|_2^2 \leq
\|A_1x_1\|_2^2 \leq \left(1+\delta_s\right) \|x_1\|_2^2$.
We have thus arrived at the conditions of Definition~\ref{defpartialRIP}.
\end{proof}

\begin{corollary}
Let $A=(A_1,A_2)$ satisfy the
RIP  of order $s$  with the RIP constant $\delta_s$. Then $A$ satisfies
 partial RIP of order $s-r$ with $\delta_{s-r}^r=\delta_s$ for partially sparse recovery of size $N-r$, for any  $r\leq s$.
\end{corollary}

\section{Partial (and total) compressibility recovery with noisy measurements}

In most realistic applications the observed measurement vector $y$
often contains noise and the true signal vector $\bar x$ is
not sparse
but rather compressible, meaning that most components
are very small but not necessarily zero. It is known, however, that Compressed Sensing is robust to noise and can
approximately recover compressible vectors.
This statement is formalized in the following theorem taken from~\cite{EJCandes_2009}.

\begin{theorem}\label{th:noisyclassical}
  Assume that the matrix $A\in\RR^{k\times N}$ satisfies  RIP with
 the RIP constant $\delta_{2s}$ such that
\(
 \delta_{2s} \; < \; \sqrt{2} - 1.
\)
For any $\bar{x} \in\RR^N$, let noisy measurements $y=A \bar{x} +\epsilon$ be given
satisfying $\| \epsilon \|_2\leq\eta$. Let $x^\#$ be a solution of
\begin{equation}\label{noisyl1_app}
\min_{x \in\RR^N}\|x\|_1 \quad \mbox{s.t.} \quad \|Ax-y\|_2\leq\eta.
\end{equation}
Then
\begin{equation}\label{noisyl1_result}
\|x^\# - \bar{x}\|_2 \; \leq \; c\eta + d\frac{\sigma_s(\bar x)_1}{\sqrt{s}},
\end{equation}
for constants $c,d$ only depending on the RIP constant, and where
$\sigma_s(\bar x)_1=\min_{x:\, \|x\|_0\leq s}\|x-\bar{x}\|_1$.
\end{theorem}

The following theorem provides an analogous result for the partially
 sparse recovery setting introduced
in Section~\ref{sec:partialCS}.

\begin{theorem} \label{th:noisy_partial}
Assume that the matrix $A=\left(A_1,A_2\right)\in\RR^{k\times N}$
satisfies partial  RIP of order $2(s-r)$ for recovery of size $N-r$ with the RIP constant $\delta_{2(s-r)}^r<\sqrt{2}-1$.
% for which~(\ref{cond_delta_noisy}) holds.
For any $\bar{x}=(\bar{x}_1,\bar{x}_2)\in\RR^{N}$,
let noisy measurements $y=A \bar{x}+\epsilon$ be given satisfying $\|\epsilon \|_2\leq\eta$. Let $x^\ast=(x_1^\ast,x_2^\ast)$ be
a solution of
\begin{equation}\label{noisyl1_partial}
\min_{x=(x_1,x_2)\in\RR^{N}}\|x_1\|_1 \quad \text{s.t.} \quad \|Ax-y\|_2\leq\eta.
\end{equation}
Then
\begin{equation}\label{noisyl1_resultpart}
\|x_1^\ast - \bar{x}_1 \|_2 \; \leq \; c\eta + d\frac{\sigma_{s-r}(\bar{x}_1)_1}{\sqrt{s-r}},
\end{equation}
and
\begin{equation} \label{noisyl1_resultpart_2}
\|x_2^\ast - \bar{x}_2\|_2 \; \leq \; C_2 \left(2\eta+ C_1 \left(c\eta + d\frac{\sigma_{s-r}(\bar{x}_1)_1}{\sqrt{s-r}}\right)\right),
\end{equation}
for constants $c,d$ only depending on $\delta_{2(s-r)}^r$, and where $C_1$ and $C_2$ are given by
\(
C_1  =  \|A_1\|_2, %  =  \max_{x_1 \neq 0} \frac{\|A_1x_1\|_2}{\|x_1\|_2},
\)
and
\(
C_2  =  \|A_2^\dagger\|_2, %  =  \frac{1}{\min_{x_2 \neq 0} \frac{\|A_2x\|_2}{\|x\|_2}}.
\)
(Since $A_2$ is full column rank recall that $A_2^\dagger = (A_2^\top A_2)^{-1} A_2^\top$ and $C_2 > 0$.)
\end{theorem}
\begin{proof}
From Theorem~\ref{th:RIP-partial-total}, the matrix $\mathcal{P}A_1$, where $\mathcal{P}$ is given by~(\ref{projM}),
satisfies the condition of Theorem~\ref{th:noisyclassical}. Thus,
since $\mathcal{P}$ is a projection matrix,
$
\| \mathcal{P} A_1 \bar{x}_1 - \mathcal{P}y\| = \| \mathcal{P} A \bar{x} - \mathcal{P}y\|
 \leq \| A \bar{x} - y \| \; \leq \; \eta,
$
and a solution~$x_1^\#$ of
\begin{equation}\label{noisyl1_proj}
\min_{x_1\in\RR^{N-r}}\|x_1\|_1 \quad \text{s.t.} \quad \| \mathcal{P} A_1x_1 - \mathcal{P}y\|_2\leq\eta,
\end{equation}
satisfies
\begin{equation}\label{noisyl1_proj_result}
\|x_1^\# - \bar{x}_1\|_2 \; \leq \; c\eta + d\frac{\sigma_{s-r}(x_1)_1}{\sqrt{s-r}}.
\end{equation}

Now, we will prove that the solutions of problems~(\ref{noisyl1_partial}) and~(\ref{noisyl1_proj}) coincide
in their $x_1$ parts,
completing thus the proof of~(\ref{noisyl1_resultpart}).
Let $(x^*_1,x^\ast_2)$ be a feasible point of~(\ref{noisyl1_partial}). Again,
since $\mathcal{P}$ is a projection matrix, we obtain that
\begin{eqnarray*}
   \|\mathcal{P} A_1x^\ast_1- \mathcal{P}y\|_2  &=&
   \| \mathcal{P} ( A_1x_1^\ast+A_2x_2^\ast-y)\|_2 \\ &\leq&  \|A_1x_1^\ast+A_2x_2^\ast-y\|_2 \;\;\; \leq \;\;\; \eta,
\end{eqnarray*}
which proves that $x^\ast_1$ is a feasible point of~(\ref{noisyl1_proj}). Now let  $x^\#_1$ be a feasible point of (\ref{noisyl1_proj}). Since $I- \mathcal{P}$ projects (orthogonally)
onto the column space of~$A_2$ there must exist an $x_2^\#$ such that
$ A_2x_2^\#=(I-\mathcal{P})(y-A_1x_1^\#)$, and then
$\|A_1x_1^\#+A_2x_2^\#-y\|_2 = \| \mathcal{P} A_1x_1^\#- \mathcal{P} y\|_2 \leq \eta$.
Therefore  $(x^\#_1,x^\#_2)$ is a feasible point of~(\ref{noisyl1_partial}).
Hence we have proved that, any solution of problem~(\ref{noisyl1_partial}) is also a solution of problem~(\ref{noisyl1_proj}),
and the inequality~(\ref{noisyl1_resultpart}) results directly from~(\ref{noisyl1_proj_result}).

We now use this inequality to bound the error on the reconstruction of $\bar{x}_2$.
Since both $\bar x$ and $x^\ast$ satisfy the measurements constraints $\| Ax - y \|_2 \leq \eta$  we have that
$\|A_1(\bar{x}_1^\ast-x_1)+A_2(\bar{x}_2^\ast-x_2)\|_2 \leq 2\eta,$
and thus $\|A_2(x_2^\ast-\bar{x}_2)\|_2 \leq 2\eta+\|A_1(x_1^\ast-\bar{x}_1)\|_2$.
Using the definitions of $C_1$ and $C_2$ we have
$\|x_2^\ast-\bar{x}_2\|_2 \leq C_2 \left( 2\eta+C_1\|x_1^\ast-\bar{x}_1\|_2 \right)$,
and the result~(\ref{noisyl1_resultpart_2}) follows from bounding $\|x_1^\ast-\bar{x}_1\|_2$
by~(\ref{noisyl1_resultpart}) in this last inequality.
\end{proof}

The condition on the matrix $A$ imposed in the previous theorem involved only its partial
RIP constant.
In the next proposition we describe how one can bound the constants $C_1$ and $C_2$ in terms of the RIP constant of $A$ (the proof is simple and is omitted, see also \cite{COSAMP}).

\begin{proposition}
Consider the RIP constant~$\delta_{s}$ of order~$s$ of $A=\left(A_1,A_2\right)\in\RR^{k\times N}$. The constants $C_1$ and $C_2$ of Theorem~\ref{th:noisy_partial} satisfy
\(
C_1 \; \leq \; \sqrt{1+\delta_{s}}
\)
and
\(
C_2 \; \leq \; \frac{1}{\sqrt{1-\delta_{s}}}.
\)
\end{proposition}

\section{Matrices with Partial RIP}

In this section we investigate regimes of $N$, $s$, and $k$ for which random Gaussian matrices satisfy partial RIP.
Similar results can be obtained for other families of random matrices, like sub-Gaussian or Bernoulli matrices.

\begin{theorem}
Let $0<\delta<1$ and $r\leq s$. Let $A =(A_1,A_2)$ with $A_1\in\RR^{k\times (N-r)}$
and $A_2 \in \RR^{k\times r}$ have independent Gaussian entries
with variance $1/k$. Then, %there exists a constant $C$ depending only on $\delta$ such that, 
as long as
\begin{eqnarray}
k   > \frac{2\times48}{3\delta^2 - \delta^3}\left( (s-r)\log\left(\frac{N-r}{s-r}e\right) + s \log\left(\frac{12}{\delta}\right)\right), \label{boundbelowforkwithC}
\end{eqnarray}
%\begin{eqnarray}%\label{boundbelowforkwithC}
%k & >& 2\frac{48}{3\delta^2 - \delta^3}\left( (s-r)\log\left(\frac{N-r}{s-r}e\right) + s \log\left(\frac{12}{\delta}\right)\right), \nonumber \\
% & = & \OOO_\delta\left( (s-r)\log\left(\frac{N-r}{s-r}\right) + s \right),\label{boundbelowforkwithC}
%\end{eqnarray}
$A=(A_1,A_2)$ satisfies partial RIP of order $s-r$ with $\delta_{s-r}^r\leq\delta$ for partially sparse recovery of size $N-r$, with high probability.
\end{theorem}

\begin{proof}
Given a particular sparsity pattern, the probability that (\ref{ineqThRIP}) does not hold is
(see~\cite[Lemma 5.1]{RBaraniuk_MDavenport_RDeVore_MWakin_2008})
\[
\leq \; 2\left(12/\delta\right)^se^{-\left(\frac{\delta^2}{16}-\frac{\delta^3}{48}\right)k}.
\]
There are ${N-r \choose s-r} \leq \left(\frac{N-r}{s-r}e\right)^{s-r}$ different sparsity patterns (see, e.g., \cite{RBaraniuk_MDavenport_RDeVore_MWakin_2008}). Let  $\mathcal{P}$ denote the probability that $A=(A_1,A_2)$ does not satisfy the partial RIP of order $s-r$ with $\delta_{s-r}^r=\delta$ for partially sparse recovery of size $N-r$. For this to happen, (\ref{ineqThRIP}) has to fail for at least one sparsity pattern, setting $\beta = \frac{\delta^2}{16}-\frac{\delta^3}{48}$ and using a union bound
\begin{eqnarray*}
\mathcal{P} &\leq &e^{(s-r)\log\left(\frac{N-r}{s-r}e\right)}2\left(\frac{12}{\delta}\right)^se^{-\beta k} \\
&\leq &2e^{\left((s-r)\log\left(\frac{N-r}{s-r}e\right)+s\log\left(\frac{12}{\delta}\right)-\beta k\right)} \\
&\leq &2e^{-\beta\left[k - \frac1\beta\left((s-r)\log\left(\frac{N-r}{s-r}e\right)+s\log\left(\frac{12}{\delta}\right)\right)\right]} \\
&\leq &2e^{-\left[(s-r)\log\left(\frac{N-r}{s-r}e\right)+s\log\left(\frac{12}{\delta}\right)\right]}, \\
&\leq &2\left(\frac{N-r}{s-r}e\right)^{-(s-r)}\left(\frac{12}{\delta}\right)^{-s},
\end{eqnarray*}
where the second to last inequality was obtained using (\ref{boundbelowforkwithC}).
%Setting, in (\ref{boundbelowforkwithC}), $C = 2\frac1\beta\log\left(\frac{12}{\delta}\right)$, and noting that $\log{N-r \choose s-r}\leq (s-r)\left(1+\log\left(\frac{N-r}{s-r}\right)\right)$, we obtain
%$$\mathcal{P} \; \leq \; e^{-\frac{C}3k} \; < \;
%e^{-\frac{C^2}3\left(s + (s-r)\log\left(\frac{N-r}{s-r}\right)\right)}.$$
It is easy to see that either $\left(e(N-r)/(s-r)\right)^{-(s-r)}$ or $\left(\frac{12}{\delta}\right)^{-s}$ goes to zero polynomially with $N$, thus $\mathcal{P}\leq \OOO\left(N^{-\OOO(1)}\right)$
%By considering $s\leq\log(N)$ and $s>\log(N)$ separately, one can see that
%there exists $\gamma>0$ such that $\mathcal{P}\leq \OOO\left(N^{-\gamma}\right)$.
\end{proof}

Note that the condition (\ref{boundbelowforkwithC}) can be asymptotically smaller than the one found in the classical case $r=0$. If, e.g., $s-r=\OOO(1)$ then $(\ref{boundbelowforkwithC})$ just requires $k = \OOO(s + \log(N-r))$ instead of the classical $k = \OOO(s \log(N/s))$.

\section{Concluding Remarks}

In some applications of Compressed Sensing one may be interested in
a  sparse (or compressible) vector whose support is  partially known in advance.
In such a setting we show that one can
 consider the $\ell_1$-minimization of the part of the vector for which
the support is not known. We have shown that such a sparse recovery can be then
 ensured under conditions that are potentially weaker than those assumed for the
 full approach.
 We have explored this feature to show that
it is possible to guarantee partial sparse recovery (with Gaussian random matrices)
for an order of measurements below the one necessary for general recovery.

\section*{Acknowledgments}
We would like to thank Rachel Ward (Math. Dept.,
UT at Austin) for interesting discussions on the topic of this paper. We also acknowledge the referees for helping us improve the paper.

\bibliographystyle{IEEEtran}

\bibliography{IEEEabrv,ref-dfol1}

\end{document}